\documentclass[oneside,10pt]{article}
\usepackage{amsmath,amssymb,amsthm,amsfonts,amsbsy,latexsym} %8411035
\usepackage{dsfont}
\usepackage{graphicx}% Include figure files

\newcommand{\bdm}{\begin{displaymath}}
\newcommand{\edm}{\end{displaymath}}
\newcommand{\bea}{\begin{eqnarray}}
\newcommand{\eea}{\end{eqnarray}}
\newcommand{\beas}{\begin{eqnarray*}}
\newcommand{\eeas}{\end{eqnarray*}}
\newcommand{\bay}{\begin{array}{c}}
\newcommand{\eay}{\end{array}}
\newcommand{\ben}{\begin{enumerate}}
\newcommand{\een}{\end{enumerate}}
\newcommand{\be}{\begin{equation}}
\newcommand{\ee}{\end{equation}}

\newcommand{\llaa}{\left\langle\hspace{-0.2cm}\left\langle}
\newcommand{\rraa}{\right\rangle\hspace{-0.2cm}\right\rangle}
\newcommand{\laa}{\langle\hspace{-0.08cm}\langle}
\newcommand{\raa}{\rangle\hspace{-0.08cm}\rangle}

\newcommand{\LZ}{L^2(\mathbb{R}^3\to\mathbb{C})}
\newcommand{\LZN}{L^2(\mathbb{R}^{3N}\to\mathbb{C})}
\newcommand{\LE}{L^1(\mathbb{R}^3\to\mathbb{C})}

\newcommand{\landau}{\mbox{\begin{scriptsize}$\mathcal{O}$\end{scriptsize}}}
\newcommand{\mlf}{\widehat{m}^{\lambda,\phi}}

\newtheorem{theorem}{Theorem}[section]
\newtheorem{lemma}[theorem]{Lemma}
\newtheorem{assumption}[theorem]{Assumption}
\newtheorem*{notation}{Notation}

\newtheorem*{remark}{Remark}
\newtheorem{definition}[theorem] {Definition}
\newtheorem{proposition}[theorem]{Proposition}

\renewcommand{\phi}{\varphi}

\begin{document}

\title{Derivation of the time dependent Gross-Pitaevskii equation without positivity condition on the interaction}

\author{Peter Pickl\footnote{
%\email{pickl@math.lmu.de}
Institute of Theoretical Physics, ETH H\"onggerberg, CH-8093
Z\"urich, Switzerland}}

\maketitle

\begin{abstract}

Using a new method \cite{pickl} it is possible to derive mean field equations from the microscopic $N$ body Schr\"odinger evolution of interacting particles
 without using BBGKY
 hierarchies. 

In this paper we wish to analyze scalings which lead to the Gross-Pitaevskii equation which is usually derived assuming positivity
of the interaction \cite{erdos1,erdos2}. The new method for dealing with mean field limits presented in \cite{pickl} allows us to relax this
condition. The price we have to pay for this relaxation is however that we have to restrict the scaling behavior to $\beta<1/3$
and that we have to assume fast convergence of the reduced one particle marginal density matrix of the initial wave function $\mu^{\Psi_0}$ 
to a pure state $|\phi_0\rangle\langle\phi_0|$.
\end{abstract}

\section{Introduction}

We are interested in solutions of the $N$-particle Schr\"odinger
equation \be\label{schroe}
i\dot \Psi_N^t = H_N\Psi_N^t \ee with symmetric $\Psi_N^0$ we
shall specify below and the Hamiltonian \be\label{hamiltonian}
 H_N=-\sum_{j=1}^N \Delta_j+\sum_{1\leq j< k\leq N} v_N^\beta(x_j-x_k) +\sum _{j=1}^N A^t(x_j)
 \ee
acting on the Hilbert space $\LZN$. $\beta\in\mathbb{R}$
stands for the scaling behavior of the interaction. The $v_N^\beta$
we wish to analyze scale with the particle number in such a way that the 
interaction energy per particle is of order one. We choose an interaction which is given
by 
\begin{assumption}\label{v}
$$v_N^\beta(x)= N^{-1+3\beta} v(N^\beta x)$$ with compactly supported, spherically symmetric $v\in L^\infty$. 
\end{assumption}

The trap potential $A^t$ 
does not depend on $N$.  $H_N$ conserves symmetry, i.e. 
any symmetric function $\Psi_N^0$ evolves into a symmetric function
$\Psi_N^t$.

Assume that the initial wave functions $\Psi_N^0\approx\prod_{j=1}^N\phi^t(x_j)$ where $\phi^0\in L^2$ and that the Gross-Pitaevskii equation
\be\label{meanfield}i\dot
\phi^t=-\left(\Delta+A^t+a|\phi^t|^2\right)\phi^t\ee with $a=\int v(x)d^3x$ has a
solution. We
shall show that  also
$\Psi_N^t\approx\prod_{j=1}^N\phi^t(x_j)$ as $N\to\infty$.

The focus of this paper is on interactions which need not be positive. The price we have to pay  is that we have to assume comparably fast convergence of the reduced one particle marginal density matrix of 
the initial wave function $\mu^{\Psi_0}$ 
to a pure state $|\phi_0\rangle\phi_0|$. Furthermore we have to restrict the scaling 
behavior of the interaction to $\beta<1/3$. 

As it seems one needs these assumptions not only for technical reasons. Without them
there might be regimes where clustering of the particle leads to a break down of the Gross-Pitaevskii description. It is clear that such a clustering can be avoided
by assuming a high purity of the condensate (i.e. fast convergence of $\mu^{\Psi_0}$ 
to $|\phi_0\rangle\phi_0|$) and moderate scaling behavior of the interaction.

\section{Counting the bad particles}

We wish to control the number of bad particles in the condensate (i.e. the particles not in the state $\phi^t$) using the method presented in \cite{pickl}.
Following \cite{pickl} we need to define some projectors first which we will do next. We shall also give some general properties of these projectors before turning to the special case of deriving the Gross-Pitaevskii equation.
\begin{definition}\label{defpro}
Let $\phi\in\LZ$.
\begin{enumerate}
\item For any $1\leq j\leq N$ the
projectors $p_j^\phi:\LZN\to\LZN$ and $q_j^\phi:\LZN\to\LZN$ are given by
\beas p_j^\phi\Psi_N=\phi(x_j)\int\phi^*(x_j)\Psi_N(x_1,\ldots,x_N)d^3x_j\;\;\;\forall\;\Psi_N\in\LZN
\eeas
and $q_j^\phi=1-p_j^\phi$. 

We shall also use the bra-ket notation
$p_j^\phi=|\phi(x_j)\rangle\langle\phi(x_j)|$.
\item 
For any $0\leq k\leq j\leq N$ we define the set $$\mathcal{A}_k^j:=\{(a_1,a_2,\ldots,a_j): a_l\in\{0,1\}\;;\;
\sum_{l=1}^j a_l=k\}$$ and the orthogonal projector $P_{j,k}^\phi$
acting on $\LZN$ as
$$P_{j,k}^\phi:=\sum_{a\in\mathcal{A}_k^j}\prod_{l=1}^j\big(p_{N-j+l}^{\phi}\big)^{1-a_l} \big(q_{N-j+l}^{\phi}\big)^{a_l}$$
and denote the special case $j=N$ by $P_k^\phi:=P_{N,k}^\phi$. For
negative $k$ and $k>j$ we set $P_{j,k}^\phi:=0$.
\item
For any function $f:\{0,1,\ldots,N\}\to\mathbb{R}^+_0$ we define the
operator $\widehat{f}^{\phi}:\LZN\to\LZN$ as
\be\label{hut}\widehat{f}^{\phi}:=\sum_{j=0}^{N} f(j)P_j^\phi\;.\ee
We shall also need the shifted operators
$\widehat{f}^{\phi}_d:\LZN\to\LZN$ given by
$$\widehat{f}^{\phi}_d:=\sum_{j=d}^{N+d} f(j+d)P_j^\phi\;.$$
\end{enumerate}
\end{definition}
\begin{notation}
Throughout the paper hats $\;\widehat{\cdot}\;$ shall solemnly be
used in the sense of Definition \ref{defpro} (c). The label $n$ shall always be used for the function $n(k)=\sqrt{k/N}$.
\end{notation}
With Definition \ref{defpro} we arrive directly at the following Lemma
based on combinatorics of the $p_j^\phi$ and $q_j^\phi$:
\begin{lemma}\label{kombinatorik}
\begin{enumerate}
\item For any functions $f,g\{0,1,\ldots,N\}\to\mathbb{R}^+_0$ we have
that
$$\widehat{f}\widehat{g}=\widehat{fg}=\widehat{g}\widehat{f}\;\;\;\;\;\;\;\;\;\;\widehat{f}p_j=p_j\widehat{f}\;\;\;\;\;\;\;\;\;\;\widehat{f}P_{j,k}=P_{j,k}\widehat{f}\;.$$
\item Let $n:\{0,1,\ldots,N\}\to\mathbb{R}^+_0$ be given by $n(k):=\sqrt{k/N}$.
Then the square of $\widehat{n}^{\phi}$ (c.f. (\ref{hut}))
equals the relative particle number operator of particles not in the
state $\phi$, i.e.
$$\left(\widehat{n}^{\phi}\right)^2=N^{-1}\sum_{j=1}^Nq_j^\phi\;.$$
\item For any $f:\{0,1,\ldots,N\}\to\mathbb{R}^+_0$ and any symmetric $\Psi\in\LZN$  \bea\label{komb1}
\left\| \widehat{f}^\phi q^\phi_1\Psi\right\|^2&=&
\|\widehat{f}^\phi\widehat{n}^{\phi}\Psi\|^2\\
\label{komb2} \left\| \widehat{f}^\phi
q^\phi_1q^\phi_2\Psi\right\|^2&\leq&
\frac{N}{N-1}\|\widehat{f}^\phi(\widehat{n}^{\phi})^2\Psi\|^2\;.\eea
\item For any function $f:\{0,1,\ldots,N\}\to\mathbb{R}^+_0$, any function $v:\mathbb{R}^6\to\mathbb{R}$ and any
$j,k=0,1,2$ we have $$\widehat{f}^\phi Q^\phi_j v(x_1,x_2)Q^\phi_k=
Q^\phi_j v(x_1,x_2)\widehat{f}^\phi_{j-k}Q^\phi_k\;,$$ where
$Q^\phi_0:=p^\phi_1 p^\phi_2$, $Q^\phi_1:=p^\phi_1q^\phi_2$ and
$Q^\phi_2:=q^\phi_1q^\phi_2$.
\end{enumerate}
\end{lemma}
\begin{proof}
(a) follows immediately from Definition \ref{defpro}, using that $p_j$
and $q_j$ are orthogonal projectors.

For (b) note that $\cup_{k=0}^N\mathcal{A}_k=\{0,1\}^N$, so $1=\sum_{k=0}^N P_k^\phi$. Using also
$(q_k^\phi)^2=q_k^\phi$ and $q_k^\phi p_k^\phi=0$ we get \beas
N^{-1}\sum_{k=1}^Nq_k^\phi&=&N^{-1}\sum_{k=1}^Nq_k^\phi\sum_{j=0}^N
P_j^\phi= N^{-1}\sum_{j=0}^N\sum_{k=1}^Nq_k^\phi
P_j^\phi=N^{-1}\sum_{j=0}^Nj P_j^\phi\eeas and (b) follows.

Let $\laa\cdot,\cdot\raa$ be the scalar product on $\LZN$. For (\ref{komb1}) we can write using symmetry of $\Psi$
\beas\|\widehat{f}^\phi\widehat{n}^{\phi}\Psi\|^2
&=&\laa\Psi,(\widehat{f}^\phi)^{2}(\widehat{n}^{\phi})^2\Psi\raa=N^{-j}\sum_{k=1}^N\laa\Psi,(\widehat{f}^\phi)^{2}q_k^\phi\Psi\raa
\\&=&\laa\Psi,(\widehat{f}^\phi)^{2}q_1^\phi\Psi\raa
=\laa\Psi,q_1^\phi(\widehat{f}^\phi)^{2}q_1^\phi\Psi\raa=\|(\widehat{f}^\phi)
q^{\phi}_1\Psi\|^2\;.
 \eeas
Similarly we have for (\ref{komb2}) \beas \|\widehat{f}^\phi
(\widehat{n}^{\phi})^2\Psi\|^2 &=&\laa\Psi,(\widehat{f}^\phi
)^2(\widehat{n}^{\phi})^4\Psi\raa
=N^{-2}\sum_{j,k=1}^N\laa\Psi,(\widehat{f}^\phi )^2q_j^\phi
q_k^\phi\Psi\raa
\\&=&\frac{N-1}{N}\laa\Psi,(\widehat{f}^\phi )^2q_1^\phi
q_2^\phi\Psi\raa+N^{-1}\laa\Psi,(\widehat{f}^\phi
)^2q_1^\phi\Psi\raa
\\&=&\frac{N-1}{N} \|\widehat{f}^\phi q^{\phi}_1q^{\phi}_2\Psi\|+N^{-1}\|\widehat{f}^\phi
q^{\phi}_1\Psi\|
 \eeas
and (c) follows.

Using the definitions above we have for (d) \beas \widehat{f}^\phi
Q^\phi_j v(x_1,x_2)Q^\phi_k
&=&\sum_{l=0}^N f(l)P_l^\phi Q^\phi_jv(x_1,x_2)Q^\phi_k
\\&=& \sum_{l=0}^N f(l)P_{N-2,l-j}^\phi Q^\phi_jv(x_1,x_2)Q^\phi_k
\\&=& \sum_{l=k-j}^{N+k-j}  Q^\phi_jv(x_1,x_2)f(l+j-k)P_{N-2,l-k}^\phi
Q^\phi_k
\\&&\hspace{-3cm}= \sum_{l=k-j}^{N+k-j}  Q^\phi_jv(x_1,x_2)f(l+j-k)P_{l}^\phi Q^\phi_k
=Q^\phi_j v(x_1,x_2)\widehat{f}^\phi _{j-k}Q^\phi_k\;.
 \eeas

\end{proof}

\section{Derivation of the Gross-Pitaevskii equation}
As presented in \cite{pickl} we wish to control the functional $\alpha_N:(\LZN\times\LZ\to\mathbb{R}^+_0)$
given by $$\alpha_N(\Psi,\phi)=\langle\Psi,\widehat{m}^\phi\Psi\rangle$$
for some appropriate weight $m:\{0,\ldots,N\}\to\mathbb{R}^+_0$.

As mentioned above we shall need comparably strong conditions on the ``purity'' of the initial condensate to derive the Gross-Pitaevskii equation without positivity assumption
on the interaction. 
This is encoded in the weights we shall choose below (see Definition \ref{mk}). For these weights convergence of the respective $\alpha$ is
 stronger than $\mu^\Psi\to|\phi\rangle\langle\phi|$ in operator norm (see Lemma \ref{convlem}). 

Note that we shall allow rather general interactions (even negative interactions) and that the Theorem below is useless when the solution of the Gross-Pitaevskii equation
does not behave nicely. 
There is a lot of literature on solutions of nonlinear Schr\"odinger equation (see for example \cite{ginibre}) showing that at least for positive $a=\int v(x)d^3x$
our assumptions on the solutions of the Gross-Pitaevskii  equation can be satisfied for many different setups.

\begin{definition}\label{mk}
For any $0<\lambda<1$  we define the function
$m^\lambda :\{1,\ldots,N\}\to\mathbb{R}^+_0$ given by
$$m^\lambda (k):=\left\{
          \begin{array}{ll}
            k/N^{\lambda}, & \hbox{for $k\leq N^{\lambda}$;} \\
            1, & \hbox{else.}
          \end{array}
        \right.
$$ We define for any $N\in\mathbb{N}$ the functional $\alpha_N^\lambda:\LZN\times\LZ\to\mathbb{R}^+_0$ by
$$\alpha_N^\lambda(\Psi_N,\phi):=\laa\Psi_N,\mlf \Psi_N\raa=\|(\mlf)^{1/2}\Psi_N\|^2\;.$$
\end{definition}

With these definitions we arrive at the main Theorem:

\begin{theorem}\label{theorem}
Let $0<\lambda,\beta<1$, let $v_N^\beta(x)$ satisfy assumption \ref{v}. Let $0<T\leq \infty$, let $A^t$ be a time
dependent potential. Assume that for any
$N\in\mathbb{N}$  there exists a solution of the Schr\"odinger equation $\Psi_N^t$ and a $L^\infty$ solution  of the
Gross-Pitaevskii equation (\ref{meanfield}) $\phi^t$ with
$\Delta |\phi^t|^2\in L^2$ for all  $0\leq t\leq T$. Then 
$$\alpha_N^\lambda(\Psi_N^t,\phi^t)\leq e^{\int_{0}^t C_v \|\phi^s\|_\infty^2
ds}\alpha_N^\lambda(\Psi_N^0,\phi^0)+(e^{\int_{0}^t C_v \|\phi^s\|_\infty^2
ds}-1)K^{\phi^t}N^{-\delta_\lambda}\;,$$ where $\delta_\lambda=\frac{1}{2}\max\{1-\lambda-4\beta\;,\;\;-1+\lambda+3\beta\}$, $C_v$ is some constant depending on $v$ only and $$K^\phi:=C_v\left(\|\Delta|\phi|^2\|+\|\phi\|_\infty+1\right)\|\phi\|_\infty\;.$$

\end{theorem}

The proof of the Theorem shall be given below.

\begin{remark}
For $\beta<1/3$ one can choose $\lambda$ such that $\delta_\lambda$ is negative.

\end{remark}

\subsection{Convergence of the reduced density matrix}

In \cite{pickl} Lemma 2.2 it is shown that convergence of $\alpha_N(\Psi,\phi)\to 0$ is equivalent to convergence of the reduced 
one particle marginal density to $|\phi\rangle\langle\phi|$
in trace norm for many different weights. The weights we use here are not covered by that Lemma. Since $m^\lambda(k)\geq k/N$ for all $0\leq k\leq N$ and all $0<\lambda<1$ it follows that
$\alpha_N^\lambda(\Psi,\phi)\geq\laa\Psi,\widehat{n}^2\Psi\raa$ (recall that $n(k)=\sqrt{k/N}$). It follows with Lemma 2.2 in \cite{pickl} that for all $0<\lambda<1$

$$\lim_{N\to\infty}\alpha_N^\lambda(\Psi,\phi)=0\;\;\;\Rightarrow \lim_{N\to\infty}\mu^{\Psi}\to|\phi\rangle\langle\phi|\;\;\text{in operator norm.}$$
Therefore our result implies convergence of the respective reduced one particle marginal density. To be able to formulate Theorem \ref{theorem} under conditions of 
the reduced one particle marginal density we have the following Lemma

\begin{lemma}\label{convlem}
Let $0<\lambda<1$, $\xi<0$ and let $\|\mu^{\Psi}-|\phi\rangle\langle\phi|\|_{op}=\landau(N^{\xi})$. Then
$$\alpha_N^\lambda(\Psi,\phi)=\landau(N^{1-\lambda+\xi})\;.$$

\end{lemma}

\begin{proof}
Under the assumptions of the Lemma it follows that $\langle\phi,\mu^{\Psi}\phi\rangle=\landau(N^{\xi})$.
Writing\beas
\mu^{\Psi_N}&=&\int
\Psi_N(\cdot,x_2,\ldots,x_N)\Psi_N^*(\cdot,x_2,\ldots,x_N)d^{3N-3}x
\nonumber\\&=&\int
p_1^\phi\Psi_N(\cdot,x_2,\ldots,x_N)p_1^\phi\Psi_N^*(\cdot,x_2,\ldots,x_N)d^{3N-3}x
\\\nonumber&&+\int
q_1^\phi\Psi_N(\cdot,x_2,\ldots,x_N)p_1^\phi\Psi_N^*(\cdot,x_2,\ldots,x_N)d^{3N-3}x
\\\nonumber&&+\int
p_1^\phi\Psi_N(\cdot,x_2,\ldots,x_N)q_1^\phi\Psi_N^*(\cdot,x_2,\ldots,x_N)d^{3N-3}x
\\\nonumber&&+\int
q_1^\phi\Psi_N(\cdot,x_2,\ldots,x_N)q_1^\phi\Psi_N^*(\cdot,x_2,\ldots,x_N)d^{3N-3}x\eeas
and using $q_1^\phi\phi(x_1)=0$ it follows that
$\|p_1^\phi\Psi_N\|^2-1=\landau(N^{\xi})$. Using $p_1^\phi+q_1^\phi=1$ and Lemma \ref{kombinatorik} (c) 
$$\|q_1^\phi\Psi_N\|^2=\laa\Psi,\widehat{n}^2\Psi\raa=\llaa\Psi,\sum_{k=0}^N\frac{k}{N}P^\phi_k\Psi\rraa=\landau(N^{\xi})\;.$$
Since  $m^{\lambda}(k)\leq N^{1-\lambda} k/N $ for any $0\leq k \leq N$ it follows that $$\alpha_N^\lambda(\Psi,\phi)\leq N^{1-\lambda}\llaa\Psi,\sum_{k=0}^N\frac{k}{N}P^\phi_k\Psi\rraa=\landau(N^{1-\lambda+\xi})\;.$$

\end{proof}

\subsection{Proof of the Theorem}

In our estimates below we shall need from time to time the operator norm $\|\cdot\|_{op}$ defined for any linear operator $f:\LZN\to\LZN$ by
$$\|f\|_{op}:=\sup_{\|\Psi\|=1}\|f\Psi\|\;.$$ In particular we shall need the following Proposition

\begin{proposition}\label{propo}
\begin{itemize}\item[(a)]
For any $f\in \LZ$
$$\|f(x_1-x_2) p_1^\phi\|_{op}\leq \|\phi\|_\infty \|f\|\;.$$
\item[(b)]For any $g\in \LE$
$$\|p_1^\phi g(x_1-x_2) p_1^\phi\|_{op}\leq \|\phi\|_\infty^2 \|g\|_1\;.$$
\end{itemize}

\end{proposition}
\begin{proof}

(a):
 Let $f\in \LZ$. Using the notation $p_1^\phi=|\phi(x_1)\rangle\langle\phi(x_1)|$
\beas
\|f(x_1-x_2) p_1^\phi\|_{op}^2&=&\sup_{\|\Psi\|=1}\|f(x_1-x_2) p_1^\phi\Psi\|^2
\\&=&\sup_{\|\Psi\|=1}\laa\Psi,|\phi(x_1)\rangle\langle\phi(x_1)|f(x_1-x_2)^2|\phi(x_1)\rangle\langle\phi(x_1)|
\Psi\raa\;.
\eeas
Using that $$\sup_{x_2\in\mathbb{R}^3}\langle\phi(x_1)|f(x_1-x_2)^2|\phi(x_1)\rangle\leq
\|\phi\|^2_\infty \|f\|^2
$$ and Cauchy Schwarz one gets
$$\|f(x_1-x_2) p_1^\phi\|_{op}^2\leq \sup_{\|\Psi\|=1} \|\Psi\|^2 \|\phi\|^2_\infty \|f\|^2\;.$$
(b): Let $g\in \LE$. 
\beas\|p_1^\phi g(x_1-x_2) p_1^\phi\|_{op}&\leq& \|p_1^\phi |g(x_1-x_2)| p_1^\phi\|_{op}
\\&=&\|p_1^\phi \sqrt{|g(x_1-x_2)|} \sqrt{|g(x_1-x_2)|}p_1^\phi\|_{op}
\\&\leq&\|\sqrt{|g(x_1-x_2)|}p_1^\phi\|_{op}^2\;.
\eeas
With (a) we get (b).

\end{proof}

We prove the Theorem using a Gr\o nwall argument. Therefore we estimate $\dot\alpha^\lambda_N(\Psi_N^t,\phi^t)$ in terms of $\alpha^\lambda_N(\Psi_N^t,\phi^t)$.
To get the estimates as stated in Theorem \ref{theorem} we need to show that
\be\label{haupt}
|\dot\alpha^\lambda_N(\Psi_N^t,\phi^t)|\leq C_v\|\phi^t\|_\infty^2\alpha^\lambda_N(\Psi_N^t,\phi^t)+K^{\phi^t}N^{-\delta}\;.
\ee
To shorten notation we use the following definitions:
\begin{definition}\label{defalpha}
Let
$$h_{j,k}:=N(N-1)v_N^{\beta}(x_j-x_k)-aN|\phi|^2(x_j)-aN|\phi|^2(x_k)\;.$$
We define the functional $\gamma^\lambda_N:\LZN\to\mathbb{R}$ by
\beas \gamma^\lambda_N(\Psi,\phi)%
&=&2\Im\left(\laa\Psi
,(\mlf_{-1}-\mlf)p_1q_2h_{1,2}p_1p_2 \Psi\raa\right)
\\&&+\Im\left(\laa\Psi ,q_1q_2h_{1,2}
(\mlf -\mlf_2)p_1p_2 \Psi\raa\right)\\&&+
2\Im\left(\laa\Psi
,(\mlf_{-1}-\mlf)q_1q_2h_{1,2}
p_1q_2 \Psi\raa\right) \;.\eeas
\end{definition}
$\gamma^\lambda_N$ was defined in such a way that for any solution of the Schr\"odinger equation $\Psi_N^t$ and any
solution $\phi^t$ of the Gross-Pitaevskii equation 
$\dot \alpha^\lambda_N(\Psi_N^t,\phi^t)=\gamma^\lambda_N(\Psi_N^t,\phi^t)$ (see Lemma \ref{ableitung} below). 
It is left to show that $\gamma^\lambda_N(\Psi_N^t,\phi^t)$ can be controlled by $\alpha^\lambda_N(\Psi_N^t,\phi^t)$ and $N^{-\delta}$ 
(which is done in Lemma \ref{alphaabl} below) to get (\ref{haupt}) and -- via Gr\o nwall -- the Theorem.

\begin{lemma}\label{ableitung}

For any solution of the Schr\"odinger equation $\Psi_N^t$, any
solution of the Gross-Pitaevskii equation $\phi^t$ and any $0<\lambda<1$ we have 
$$\dot \alpha^\lambda_N(\Psi_N^t,\phi^t)=\gamma^\lambda_N(\Psi_N^t,\phi^t)\;.$$
\end{lemma}

\begin{proof}
Let $$H^\phi_{GP}:=\sum_{k=1}^N -\Delta_k+A(x_k)+a|\phi|^2(x_k)$$
be the sum of Gross-Pitaevskii Hamiltonians in each particle. It follows that
\be\label{ablp}\frac{d}{dt}\widehat{f}^{\phi^t}=i[H^{\phi^t}_{GP},\widehat{f}^{\phi^t}]\ee
for any function $f:\{0,\ldots,N\}\to\mathbb{R}$. For ease of notation we shall drop now the indices $\phi$ and $\lambda$ for the rest of the proof. With (\ref{ablp}) we get
\beas\dot\alpha _N(\Psi_N^t,\phi^t)
%
%&=&\frac{d}{dt}\laa\Psi_N^t
%,\widehat{m}  \Psi_N^t\raa
%
%\\
&=&i\laa\Psi_N^t
,\widehat{m}  H\Psi_N^t\raa-i\laa H\Psi_N^t
,\widehat{m}  \Psi_N^t\raa+i\laa \Psi_N^t
,[H_{GP},\widehat{m}]\Psi_N^t\raa
\\&=&-i\laa\Psi_N^t
,[H-H_{GP},\widehat{m}]\Psi_N^t\raa\;.
 \eeas
Using symmetry of $\Psi_N^t$ and selfadjointness of $h_{j,k}$ it follows that
\bea\label{wiehier} \nonumber\dot\alpha _N(\Psi_N^t,\phi^t)%
&=&-i(N^2-N)^{-1}\sum_{1\leq j< k\leq N}\laa\Psi_N^t
,[h_{j,k},\widehat{m}  \;]\Psi_N^t\raa
\\\nonumber&=&-\frac{i}{2}\left(\laa\Psi_N^t
,h_{1,2}\widehat{m}  \;\Psi_N^t\raa-\laa\Psi_N^t
,\widehat{m}  \;h_{1,2}\Psi_N^t\raa\right)
\\&=&\Im\left(\laa\Psi_N^t
,h_{1,2}\widehat{m}  \;\Psi_N^t\raa\right)\;.
\eea Note that we can write for any
$m:\{1,\ldots,N\}\to\mathbb{R}^+_0$ (remember that $P_{N,k}=0$
whenever $k<0$ or $k>N$)  \bea\label{nersetzen}
\nonumber\widehat{m}  &=&\sum_{k=0}^N m (k)P_k
\\\nonumber &=&\sum_{k=0}^{N-2}
\big(m (k)p_1  p_2  P_{N-2,k}+m (k)p_1 q_2 P_{N-2,k-1}\\\nonumber
&&+m (k)q_1  p_2 P_{N-2,k-1}+m (k)(1-p_1 q_2-q_1  p_2-p_1 p_2)
P_{N-2,k-2}\big)
\\\nonumber &=&\sum_{k=0}^{N}
\big(m (k)p_1  p_2  P_{N-2,k}+m (k)p_1 q_2 P_{N-2,k-1}\\\nonumber
&&+m (k)q_1  p_2 P_{N-2,k-1}+m (k)P_{N-2,k-2}\big)
\\\nonumber&&-\sum_{k=0}^{N}m (k+1)p_1  q_2P_{N-2,k-1}-m (k+1)q_1  p_2P_{N-2,k-1}\\\nonumber&&-m (k+2)p_1 p_2P_{N-2,k})
\\\nonumber
&=&(\widehat{m}  -\widehat{m}  _2)p_1p_2+(\widehat{m}  -\widehat{m}  _1)p_1q_2+(\widehat{m}  -\widehat{m}  _1)q_1p_2\\&&+\sum_{k=0}^Nm (k)P_{N-2,k-2}\;.
 \eea
Using  symmetry of $\Psi_N^t$ and selfadjointness of
$h_{1,2}P_{N-2,k-2}$ it follows that
\be\label{cancel} \dot\alpha _N(\Psi_N^t,\phi^t)%
=\Im\left(\laa\Psi_N^t ,h_{1,2}
\left((\widehat{m}  -\widehat{m}  _2)p_1p_2+2(\widehat{m}  -\widehat{m}  _1)p_1q_2\right)
\Psi_N^t\raa\right) \;.\ee
Since  $1=p_1 p_2+p_1 q_2+q_1 p_2+q_1 q_2$
\beas\dot\alpha _N(\Psi_N^t,\phi^t)&=&
 \Im\left(\laa\Psi ,p_1p_2h_{1,2}
(\widehat{m}  -\widehat{m}  _2)p_1p_2 \Psi\raa\right)
\\&&+\Im\left(\laa\Psi ,p_1q_2h_{1,2}
(\widehat{m}  -\widehat{m}  _2)p_1p_2 \Psi\raa\right)
\\&&+\Im\left(\laa\Psi ,q_1p_2h_{1,2}
(\widehat{m}  -\widehat{m}  _2)p_1p_2 \Psi\raa\right)
\\&&+\Im\left(\laa\Psi ,q_1q_2h_{1,2}
(\widehat{m}  -\widehat{m}  _2)p_1p_2 \Psi\raa\right)
\\&&+
 2\Im\left(\laa\Psi ,p_1p_2h_{1,2}
(\widehat{m}  -\widehat{m}  _1)p_1q_2 \Psi\raa\right)
\\&&+2\Im\left(\laa\Psi ,p_1q_2h_{1,2}
(\widehat{m}  -\widehat{m}  _1)p_1q_2 \Psi\raa\right)
\\&&+2\Im\left(\laa\Psi ,q_1p_2h_{1,2}
(\widehat{m}  -\widehat{m}  _1)p_1q_2 \Psi\raa\right)
\\&&+2\Im\left(\laa\Psi ,q_1q_2h_{1,2}
(\widehat{m}  -\widehat{m}  _1)p_1q_2 \Psi\raa\right)\;. \eeas Using
that
$\Im(\laa\Psi,A\Psi\raa)=-\Im(\laa\Psi,A^t\Psi\raa)$ for
any operator $A$ and that $\Psi$ is symmetric (note that $p_1
q_2h_{1,2}q_1  p_2$ is invariant under adjunction with simultaneous exchange
 of the variables $x_1$ and $x_2$) and Lemma \ref{kombinatorik} (d) we
get \beas\dot\alpha _N(\Psi_N^t,\phi^t)&=&\nonumber
2\Im\left(\laa\Psi ,p_1q_2h_{1,2}
(\widehat{m}  -\widehat{m}  _2)p_1p_2 \Psi\raa\right)
\\&&-2\Im\left(\laa\Psi ,p_1q_2
(\widehat{m}  -\widehat{m}  _1)h_{1,2}p_1p_2 \Psi\raa\right)
\\&&+\Im\left(\laa\Psi ,q_1q_2h_{1,2}
(\widehat{m}  -\widehat{m}  _2)p_1p_2 \Psi\raa\right)
\\&&+2\Im\left(\laa\Psi ,q_1q_2h_{1,2}
(\widehat{m}  -\widehat{m}  _1)p_1q_2 \Psi\raa\right) \;.\eeas 
Lemma \ref{kombinatorik} (d) applied to the first and fourth summand completes the proof.

\end{proof}

With Lemma \ref{ableitung} equation (\ref{haupt}) follows once we can control the different summands appearing in $\gamma^\lambda_N$ in  a suitable way. So the following Lemma
completes the proof of the Theorem.

\begin{lemma}\label{alphaabl}
Let $v_N^\beta$ satisfy assumption \ref{v}. Then there exists a $C<\infty$ such that for any $\phi\in L^\infty$ with $\Delta|\phi|^2\in L^2$
\begin{enumerate}
\item
$$\left|\laa\Psi
,(\mlf_{-1}-\mlf)p_1q_2h_{1,2}p_1p_2
\Psi\raa\right|
\leq K^\phi N^{\delta_\lambda}$$
\item
$$\left|\laa\Psi ,q_1q_2h_{1,2}
(\mlf -\mlf_2)p_1p_2 \Psi\raa\right|
\leq C\|\phi\|_\infty^2\alpha^\lambda_N(\Psi,\phi)+ K^\phi N^{\delta_\lambda}
$$
\item
$$\left|\laa\Psi
,(\mlf_{-1}-\mlf)q_1q_2h_{1,2}
p_1q_2 \Psi\raa\right|
\leq  K^\phi N^{\delta_\lambda} $$
\end{enumerate}
with $\delta_\lambda$ and $K^\phi$ as in Theorem \ref{theorem}.

\end{lemma}

Before we prove the Lemma a few words on (a) and (c) first:
It is (a) which is physically the most important. Here the
mean
 field cancels out most of the interaction. The
central point in the mean field argument is observing that
$p_1q_2h_{1,2}p_1p_2$ is small.

For (c) the choice of the weights $m^\lambda$ plays an important role. Note that we only have one projector $p$ here and $\|q_1q_2h_{1,2}
p_1q_2\|_{op}$ can not be controlled by the $L^1$-norm of $v$ (see Proposition \ref{propo}).
On the other hand we have altogether three projectors $q$ in (c). Assuming that the condensate is very clean (which is encoded in $\widehat{m}^\lambda $) these $q$'s make (c) small.

%But $m^\lambda(k-1)-m^\lambda(k-2)$ is for large $N$ approximately the $k$-derivative of $m$, so $\widehat{m}^\lambda _1-\widehat{m}^\lambda _2\approx \widehat{k^{-1}m}$.
%On the other hand each $q$ yields a factor $\sqrt{k/N}$ (see Lemma \ref{kombinatorik} (c)). But the factor $m^\lambda(k)=0$ for all $k>M$. Thus the three projectors $q$ appearing in (c)
%can heuristically be seen as a factor $(M/N)^{3/2}$

\begin{proof}

In the proof we shall drop the index $\lambda$ and $\phi$ for ease of notation. Constants appearing in estimates will generically be denoted by $C$. We shall not distinguish constants appearing in a sequence of estimates, i.e. in $X\leq CY\leq CZ$ the constants may differ.

 In bra-ket notation $p_1=|\phi(x_1)\rangle\langle\phi(x_1)|$. Writing $\star$ for the convolution we get for any $f:\mathbb{R}^6\to\mathbb{R}$
\be\label{neu} p_1f(x_1-x_2)p_1=|\phi(x_1)\rangle\langle\phi(x_1)|f(x_1-x_2)|\phi(x_1)\rangle\langle\phi(x_1)|=p_1 (f\star|\phi|^2)(x_2)\;,\ee in particular
$$p_1\delta(x_1-x_2)p_1=p_1 |\phi(x_2)|^2\;.$$
With $p_1q_1=0$ it follows that
\beas p_1q_2h_{1,2}p_1p_2&=&Np_1q_2\left((N-1)v_N^{\beta}(x_1-x_2)-a|\phi|^2(x_2)\right)p_1p_2
\\&=&Np_1q_2\left((N-1)v_N^{\beta}(x_1-x_2)-a\delta(x_1-x_2)\right)p_1p_2\;.\eeas
Using this and triangle inequality
the left hand side of (a) is bounded by
\bea\label{firstbound}\nonumber&&N|\laa\Psi
,(\widehat{m}_{-1}-\widehat{m})p_1q_2\left(Nv_N^{\beta}(x_1-x_2)-a\delta(x_1-x_2)\right)p_1p_2
\Psi\raa|\\&&\;\;\;\;\;+\;\;N|\laa\Psi
,(\widehat{m}_{-1}-\widehat{m})p_1q_2v_N^{\beta}(x_1-x_2)p_1p_2
\Psi\raa|\;.
%\\&\leq&N\|p_1\left((N-1)v_N^{\beta}(x_1-x_2)-a|\phi|^2(x_2)\right)p_1\|_{op}\;\|(\widehat{m}  _1-\widehat{m}  _2)q_2\Psi\|\;\|p_2\Psi\|
\eea
To control the first summand we define the function $f_N^{\beta}:\mathbb{R}^3\to\mathbb{R}$ by
$$\Delta f_N^{\beta}=Nv_N^{\beta}-a\delta\;.$$
Recall that $v$ is compactly supported. Since $N\int v(x)d^3x=a$ the integration constant of $f_N^{\beta}$ can be chosen such that also $f_N^{\beta}$ has compact support. Using the scaling behavior of $v_N^{\beta}$ it follows that there exits a function $f:\mathbb{R}^3\to\mathbb{R}$ with
$$f_N^{\beta}=N^{\beta}f(N^{\beta}x)\;\;\;\;\;\;\text{and}\;\;\;\;\;\;\|f_N^{\beta}\|_1=N^{-2\beta}\|f\|_1\;.$$
Now we can estimate the first summand in (\ref{firstbound}) using (\ref{neu})
\beas
&&N|\laa\Psi
,(\widehat{m}_{-1}-\widehat{m})p_1q_2\Delta f_N^{\beta}(x_1-x_2)p_1p_2
\Psi\raa|
\\&=&N|\laa\Psi
,(\widehat{m}_{-1}-\widehat{m})p_1q_2\left((\Delta f_N^{\beta})\star |\phi|^2 \right)(x_2)p_1p_2
\Psi\raa|
\\&=&N|\laa\Psi
,(\widehat{m}_{-1}-\widehat{m})p_1q_2 \left(f_N^{\beta}\star (\Delta|\phi|^2)\right) (x_2)p_1p_2
\Psi\raa|\;.
\eeas
Since  $\|p_1p_2
\Psi\|\leq 1$ one gets with Proposition \ref{propo}
\beas
&\leq&N\|(\widehat{m}_{-1}-\widehat{m})q_2\Psi\|\;\| p_1 \left(f_N^{\beta}\star (\Delta|\phi|^2)\right) (x_2)p_1\|_{op}
\\&\leq&N\|(\widehat{m}_{-1}-\widehat{m})q_2\Psi\|\;\|f_N^{\beta}\star (\Delta|\phi|^2)\|\;\|\phi\|_\infty
\;.
\eeas
In view of Lemma \ref{defpro} (b) we have using symmetry of $\Psi$ for the first factor \bea\label{mdiffop}
\|(\widehat{m}_{-1}-\widehat{m})q_2\Psi\|
&=&\nonumber\|(\widehat{m}_{-1}-\widehat{m})\widehat{n}\Psi\|
\\&\leq&\sup_{0\leq k\leq N^\lambda } \left(\left|\frac{k-1}{N^\lambda}-\frac{k}{N^\lambda}\right|\sqrt{k/N}\right)=(N^\lambda  N)^{-1/2}\;.
\eea
Using Young's inequality we have for the second factor
\beas
\| f_N^{\beta}\star (\Delta|\phi|^2)\|\leq \| f_N^{\beta}\|_1\;\|\Delta|\phi|^2\|\:\leq\;CN^{-2\beta}\;\|\Delta|\phi|^2\|\;.
\eeas
It follows that the first summand of (\ref{firstbound}) is bounded by \be\label{bound11}
CN^{-\lambda/2}\|\Delta|\phi|^2\|\;\|\phi\|_\infty N^{1/2-2\beta}\;.
\ee
Using Schwarz inequality, then Proposition \ref{propo} and equation (\ref{mdiffop}) the second summand of (\ref{firstbound}) is smaller than
\beas
&&N\|(\widehat{m}_{-1}-\widehat{m})q_2\Psi\|\;\|p_1v_N^{\beta}(x_1-x_2)p_1\|_{op}
\\&\leq&N\|(\widehat{m}_{-1}-\widehat{m})q_2\Psi\|\;\|v_N^{\beta}\|_1\;\|\phi\|_\infty^2
\leq C(N^\lambda  N)^{-1/2}\|\phi\|_\infty^2\;.
\eeas
Using this and (\ref{bound11}) we get (a).

%-------------------------------------------------------------------------------------------------

For (b) we use first that $q_1q_2 w(x_1)p_1p_2=0$ for any function $w$.
It follows with Lemma \ref{defpro} (d) that
\bea\label{suffsym}&&\laa\Psi,q_1q_2h_{1,2}
(\widehat{m}-\widehat{m}_2)p_1p_2 \Psi\raa
\\\nonumber&=&(N^2-N)\laa\Psi
,q_1q_2(\widehat{m}  _{-2}-\widehat{m}  )^{1/2}v_N^{\beta}(x_1-x_2)
(\widehat{m}  -\widehat{m}  _2)^{1/2}p_1p_2 \Psi\raa\;.
\eea
Before we estimate this term note that the operator norm of $q_1q_2v_N^{\beta}(x_1-x_2)
$ restricted to the subspace of
symmetric functions is much smaller than the operator norm on full
$\LZN$. This comes from the fact that $v_N(x_1-x_2)$ is only nonzero in a small area where $x_1\approx x_2$. A non-symmetric wave function may be fully localized in that area, whereas for a symmetric wave function only a small part lies in that area. 
To get sufficiently good control of (\ref{suffsym}) we symmetrize $(N-1)v_N^{\beta}(x_1-x_2)$ replacing it by 
$\sum_{k=2}^Nv_N^{\beta}(x_1-x_k)$ and get
 \beas(\ref{suffsym})
&=&(N^2-N)\laa\Psi
,q_1q_2(\widehat{m}  _{-2}-\widehat{m}  )^{1/2}v_N^{\beta}(x_1-x_2)
(\widehat{m}  -\widehat{m}  _2)^{1/2}p_1p_2 \Psi\raa
\\\nonumber&=&N\laa\Psi ,(\widehat{m}  _{-2}-\widehat{m}  )^{1/2}\sum_{j=2}^Nq_1q_jv_N^{\beta}(x_1-x_j)
p_1p_j (\widehat{m}  -\widehat{m}  _2)^{1/2}\Psi\raa
%
%\\&=&\laa\Psi ,(\widehat{m}  _{-2}-\widehat{m}  )^{1/2}\sum_{j=2}^Nq_1q_jv_N^{\beta}(x_1-x_j)
%p_1p_j (\widehat{m}  -\widehat{m}  _2)^{1/2}\Psi\raa
%
\\&\leq&N\|(\widehat{m}  _{-2}-\widehat{m}  )^{1/2}q_1\Psi\|\;\|\sum_{j=2}^Nq_jv_N^{\beta}(x_1-x_j)
p_1p_j (\widehat{m}  -\widehat{m}  _2)^{1/2}\Psi\|\;. \eeas For the first
factor we have since $\left(m (k)-m (k-2)\right)k/N\leq 2N^{-1}m (k)$ in view
of Lemma \ref{kombinatorik} (c) that \beas
\|(\widehat{m}  _{-2}-\widehat{m}  )^{1/2}q_1\Psi\|^2\;=\;\laa\Psi(\widehat{m}  _{-2}-\widehat{m}  )\widehat{n}^2\Psi\raa
\;\leq\;2N^{-1}\alpha _N(\Psi,\phi)\;.
\eeas The second factor is bounded by \bea\label{zweifreunde}
%\|\sum_{j=2}^Nq_jv_N^{\beta}(x_1-x_j) p_1p_j
%(\widehat{m}  -\widehat{m}  _2)^{1/2}\Psi\|^2
%
%\\\nonumber=
&&\hspace{-0.55cm}\sum_{2\leq j< k\leq N}^N\laa(\widehat{m}  -\widehat{m}  _2)^{1/2}\Psi, p_1p_j
v_N^{\beta}(x_1-x_j)q_jq_kv_N^{\beta}(x_1-x_k)(\widehat{m}  -\widehat{m}  _2)^{1/2}p_1p_k\Psi\raa\nonumber\\&&\;\;\;+\sum_{k=2}^N\|q_kv_N^{\beta}(x_1-x_k)p_1p_k(\widehat{m}  -\widehat{m}  _2)^{1/2}\Psi\|^2\;.
 \eea
Using symmetry and Proposition \ref{propo} the first summand in (\ref{zweifreunde}) is bounded by \beas&&\hspace{-1cm}
N^2\laa(\widehat{m}  -\widehat{m}  _2)^{1/2}\Psi, p_1p_2q_3
v_N^{\beta}(x_1-x_2)v_N^{\beta}(x_1-x_3)p_1q_2p_3(\widehat{m}  -\widehat{m}  _2)^{1/2}\Psi\raa
\\&\leq&N^2\|\sqrt{|v_N^{\beta}(x_1-x_2)|}\sqrt{|v_N^{\beta}(x_1-x_3)|}p_1q_2p_3(\widehat{m}  -\widehat{m}  _2)^{1/2}\Psi\|^2
\\&\leq&N^2\|\sqrt{|v_N^{\beta}(x_1-x_2)|}p_1\|^4_{op}\;\|(\widehat{m}  -\widehat{m}  _2)^{1/2}q_2\Psi\|^2
\\&\leq&N^2\|\phi\|_\infty^4
\|v_N^{\beta}\|_1^2\|(\widehat{m}  -\widehat{m}  _2)^{1/2}q_2\Psi\|^2
\\&\leq&C\|\phi\|_\infty^4\alpha _N(\Psi,\phi)
\;. \eeas
Using Proposition \ref{propo} the second summand in (\ref{zweifreunde}) can be controlled by \beas
&&N\laa(\widehat{m}  -\widehat{m}  _2)^{1/2}\Psi, p_1p_2
(v_N^{\beta}(x_1-x_2))^2p_1p_2(\widehat{m}  -\widehat{m}  _2)^{1/2}\Psi\raa
\\&\leq& N
\|p_1(v_N^{\beta}(x_1-x_2))^2p_1 \|_{op}\;\|(\widehat{m}  -\widehat{m}  _2)^{1/2}\|_{op}^2 
\\&\leq& N\|\phi\|^2_\infty\;
\|v_N^{\beta}\|_2^2\;\|(\widehat{m}  -\widehat{m}  _2)^{1/2}\|_{op}^2 \leq
C\|\phi\|^2_\infty NN^{-2+3\beta} N^{-\lambda}\;. \eeas It follows that (b) is bounded by

\beas C\|\phi\|_\infty^2\alpha _N(\Psi,\phi)+ C\|\phi\|_\infty N^{-1/2+3/2\beta}
N^{-\lambda/2}\;.\eeas

%-----------------------------------------------------------------------------------------------------------------------------------------

Next we shall prove (c).
Using Lemma \ref{defpro} (d) and Cauchy-Schwarz we get for the left hand side of (c)
\beas
&&\hspace{-1cm}\left|\laa\Psi
,(\widehat{m}_{-1}-\widehat{m})\widehat{n}_1q_1q_2h_{1,2}
\widehat{n}^{-1}p_1q_2 \Psi\raa\right|
\\&& \hspace{1cm}\leq \|(\widehat{m}_{-1}-\widehat{m})\widehat{n}_1q_1q_2\Psi\|\;
\|h_{1,2}\widehat{n}^{-1}p_1q_2 \Psi\|\;.
\eeas
For the first factor we have using Lemma \ref{defpro} (c)
\beas
\|(\widehat{m}_{-1}-\widehat{m})\widehat{n}_1q_1q_2\Psi\|&\leq&\frac{N}{N-1}\|(\widehat{m}_{-1}-\widehat{m})\widehat{n}_1\widehat{n}^2\Psi\|
\\&\leq&\sup_{0\leq k\leq N^\lambda }\left(\frac{N}{N-1}\left|\frac{k-1}{N^\lambda }-\frac{k}{N^\lambda  }\right|\sqrt{(k+1)/N}\frac{k}{N}\right)
\\&=&\frac{\sqrt{N^\lambda +1}}{(N-1)\sqrt{N}}\;.
\eeas
For the second factor  we have using Proposition \ref{propo} and Lemma \ref{kombinatorik} (c)
\beas
\|h_{1,2}\widehat{n}^{-1}p_1q_2 \Psi\|&\leq&\|h_{1,2}p_1\|_{op}\;\|\widehat{n}^{-1}q_2 \Psi\|
\\&\leq&\|\phi\|_\infty\;\|h_{1,2}\|\leq N\|\phi\|_\infty\;\left((N-1)\|v_N^\beta\|+\|2a|\phi|^2\|\right)\;.
\eeas
Since the scaling of $v_N^\beta$ is such that $\|v_N^\beta\|=\|v\|N^{-1+3/2\beta}$  it follows that (c) is bounded by
\beas&&\hspace{-1cm}CN\frac{\sqrt{N^\lambda +1}}{(N-1)\sqrt{N}} (N-1)N^{-1+3/2\beta}( \|\phi\|_\infty+ \|\phi\|_\infty^2)\\&&\hspace{-1cm}\leq C ( \|\phi\|_\infty+ \|\phi\|_\infty^2)N^{\lambda/2}N^{-1/2+3/2\beta}\eeas
and (c) follows.
\beas
\eeas

\end{proof}

\end{document}